%% file: arxiv_Bo_Yu_Yang.tex
\documentclass[conference,letterpaper]{IEEEtran}

\input{./macro.tex}

\usepackage[utf8]{inputenc} 
\usepackage[T1]{fontenc}
\usepackage{url}
\usepackage{ifthen}
\usepackage{cite}
\usepackage[cmex10]{amsmath} 

\usepackage{amssymb, amsthm, dsfont, braket, mathtools, shadethm, xfrac}
\usepackage{paralist, fancyhdr, etoolbox}
\usepackage{algorithmic}
\usepackage{algorithm}
\usepackage{graphicx}
\usepackage{textcomp}
\usepackage{xcolor}
\usepackage{stfloats}
\usepackage{flushend}
\usepackage[hidelinks]{hyperref}
\usepackage{tikz}
\usetikzlibrary{positioning,calc,shapes,arrows}
\tikzset{
block/.style = {draw, fill=white, rectangle, minimum height=3em, minimum width=3em},
tmp/.style  = {coordinate}, 
sum/.style= {draw, fill=white, circle, node distance=1cm},
input/.style = {coordinate},
output/.style= {coordinate},
pinstyle/.style = {pin edge={to-,thin,black}
}
}

\usepackage{cases}
\usepackage{indentfirst,csquotes}
\usepackage[T1]{fontenc}
\newtheorem{theo}{Theorem}
\newtheorem{defn}{Definition}
\newtheorem{exam}{Example}
\newtheorem{lemma}{Lemma}
\newtheorem{prop}{Proposition}
\newtheorem{coro}{Corollary}
\newtheorem{remark}{Remark}
\usepackage{balance}
\usepackage[normalem]{ulem} 
\usepackage{amsmath,amsthm,amssymb,bm,bbm,dsfont,braket}
\usepackage[mathscr]{eucal}

\newcommand{\sh}{\kern-0.08em$^\textbf{\#}$\hspace{-3pt}}
\renewcommand{\b}{\kern-0.06em$\flat$}

\makeatother
\begin{document}
\title{Single-Period Portfolio Selection\\
via Information Projection} 


\author{%
  \IEEEauthorblockN{Bo-Yu Yang and Michael Gastpar}
  \IEEEauthorblockA{School of Computer and Communication Sciences\\
EPFL, Lausanne, Switzerland\\
Email: \href{mailto:bo-yu.yang@epfl.ch}{bo-yu.yang@epfl.ch}, \href{mailto:michael.gastpar@epfl.ch}{michael.gastpar@epfl.ch}}
}

\maketitle


\begin{abstract}
We study the single-period portfolio selection problem under Constant Relative Risk-Aversion (CRRA) utility through the information-theoretic lens.
Assuming only that the market payoff vector has finite support, we show that the Certainty-Equivalent (CE) growth rate under CRRA utility can be decomposed into a portfolio-induced Rényi divergence term, a Rényi entropy term of the risk-tilted market law, and a log-partition term.
In this setting, the Rényi order has a clear operational meaning: it exactly coincides with the investor's coefficient of relative risk aversion.
We further show that CRRA portfolio selection is equivalent to a Rényi information-projection problem.
Using a variational representation of Rényi divergence, we obtain a Blahut--Arimoto-style alternating optimization with a closed-form auxiliary update and a KL-type portfolio step.
In the low risk-aversion regime, this method empirically requires fewer iterations than both direct CRRA utility optimization and Cover's method.
\end{abstract}


\section{Introduction}

Maximizing reward or utility under uncertainty is a central goal across reinforcement learning \cite{Sutton1998RL,Peters2010REPS}, bandits \cite{Agrawal2013Thompson,Russo2016IT_Thompson}, and finance \cite{Cochrane2009asset,Back2017asset,Campbell2017Financial}, and it also appears in information theory through gambling \cite{Kelly1956infor,Bleuler2020bet} and portfolio theory \cite[Ch.\,16]{Cover2006}.
One natural setting where these perspectives meet is \textit{portfolio selection}: 
in a single-period model, an investor allocates wealth according to a portfolio vector \(\mb b\) before a random market payoff vector \(\mb X\) is realized.
Here \(X_i\) denotes the nonnegative payoff per unit wealth invested in asset \(i\), following Cover's price-relative convention \cite{Cover1984LogInvestment}.
For example, an asset bought at \(20\) and sold at \(30\) has payoff \(X_i=30/20=1.5\).
With unit initial wealth, a portfolio \(\mb b\) yields realized wealth \(w=\langle \mb b,\mb x\rangle\), which is evaluated by the utility function \(u(w)\).

Because future payoffs are uncertain, we evaluate portfolios through expected utility, a framework rooted in Bernoulli's 1738 treatment of risk \cite{Bernoulli1954risk} and the von Neumann--Morgenstern axiomatization of utility in the 1940s \cite{Von1947games}.

The closest information-theoretic antecedents come from gambling.
For logarithmic growth rate (c.f.~\(G_1\) in Def.~\ref{def:CE_growth}), Kelly's seminal work~\cite{Kelly1956infor} and the treatment in~\cite[Ch.\,6]{Cover2006} connect optimal growth to Shannon entropy and KL divergence.
More recently, Bleuler, Lapidoth, and Pfister \cite{Bleuler2020bet} extended this perspective in horse betting beyond logarithmic growth to a family of risk-sensitive utilities corresponding to what we call the certainty-equivalent (CE) growth rate \(G_{\rho_u}\) in Definition~\ref{def:CE_growth}, with R{\'e}nyi divergence playing the central role.
Their gambling results rely on the diagonal structure of the payoff matrix (Def.~\ref{def:payoff_matrix}), where outcomes are mutually exclusive; in contrast, general portfolio selection admits dense payoff matrices.\footnote{For example, in horse betting only one horse wins a race, whereas in a financial market multiple assets may have positive payoffs simultaneously.}
This raises the question of whether the information-theoretic picture behind gambling extends to general portfolio selection.

In the single-period CRRA portfolio selection problem with the given payoff law $p$, the optimal portfolio $\mb b^\star$ satisfies the first-order condition (FOC)
\begin{subnumcases} 
{\mathbb{E}_{p}\!\left[
\frac{\langle \mb{b}^\star,\mb{X}\rangle^{-\rho_u}} {\mathbb{E}_{p}\!\left[\langle \mb{b}^\star,\mb{X}\rangle^{1-\rho_u}\right]} X_i \right]}
= 1, & $b_i^\star > 0$,\label{eq:crra_foc}\\
\le 1, & $b_i^\star = 0$,
\end{subnumcases}
for each asset \(i\in[m]\); see, e.g., \cite[Ch.~16]{Cover2006}, \cite{Mossin1968MultiPeriod}.
This optimality condition, obtained from the KKT conditions (or Lagrange duality), is implicit in \(\mathbf b^\star\).
Nevertheless, when maximizing expected log-wealth \((\rho_u=1)\), Cover's algorithm \cite{Cover1984LogInvestment} gives the multiplicative update in~\eqref{eq:cover_update}, which monotonically improves the expected log-wealth, with objective values converging to the global optimum.
Rather than directly extending Cover's multiplicatively updating algorithm to CRRA utility, we discover another form of ``duality'' in Theorem~\ref{thm:renyi_projection}, where the maximizer $\mb b^\star$ is exactly the minimizer in the information projection problem.

In Csisz{\'a}r's \(I\)-divergence geometry, a reference law is projected onto a convex class by minimizing KL divergence, yielding Pythagorean-type characterizations and alternating-minimization procedures~\cite{Csiszar1975IDivergence,CsiszarTusnady1984InforGeo}.
R{\'e}nyi projection theory develops an analogous picture on \(\alpha\)-convex sets and \(\alpha\)-linear families~\cite{Kumar2016RenyiProjection}.
Our projection has a different origin: the target family is not imposed through moment or marginal constraints, but is generated by the payoff matrix and the feasible portfolios.
After symmetrization, CRRA portfolio selection projects the risk-tilted market law \(\tilde p\), under R{\'e}nyi divergence, onto the attainable normalized wealth laws
\(\{\bar q_{\mb b}:\mb b\in\Delta^{m-1}\}\), where \(\bar q_{\mb b}\) is obtained by normalizing the portfolio-induced wealth \(q_{\mb b}(\mb x)=\langle \mb b,\mb x\rangle\).

The contribution of this work is structural.
Theorem~\ref{thm:G_rho} decomposes the CE growth rate on a finite symmetric covering of the payoff support into a portfolio-induced R{\'e}nyi divergence, a R{\'e}nyi entropy term, and a portfolio-independent log-partition term.
Since every finite payoff support admits such a covering, this yields Theorem~\ref{thm:renyi_projection}: CRRA portfolio selection is equivalent to a R{\'e}nyi information-projection problem, with R{\'e}nyi order equal to the investor's relative risk aversion.
Theorem~\ref{thm:renyi_variational} then converts this projection problem into a Blahut--Arimoto-style alternating optimization with a closed-form auxiliary update and a convex portfolio step.
The resulting information-projection exponentiated-gradient (Info-Proj EG) method empirically reaches smaller optimization error in fewer accepted first-order iterations in low-risk-aversion regimes, compared with naively applying EG to the original objective and with Cover's method.

\section{Problem Formulation}
\subsection{Notation}
Random vectors are denoted by uppercase letters and their realizations by lowercase letters.
Vectors are boldface, inner products are written as \(\langle \mb b,\mb x\rangle\).
We write \(\Delta^{m-1}:=\{\mb b\in\mathbb{R}_+^m:\langle \mb b,\mds{1}\rangle=1\}\) for the portfolio simplex, and \(\Delta(\mcal A)\) for the probability simplex on a finite alphabet \(\mcal A\).
Unless otherwise stated, asset indices use \(i\in[m]\), while market-state indices use \(j\in[k]\).
The \(j\)-th entry of a vector \(\mathbf{x}\) is denoted by \(x_j\).
We write \(p,q,r\) for probability measures, \(\mathbb E_p[\,\cdot\,]\) for expectation, and \(\supp(\cdot)\) for support. 
We denote the all-one vector by $\mds{1}$.

\subsection{Preliminaries}

Consider a rational investor who chooses a portfolio to maximize the expected utility $\mds{E}[u(W)]$ of wealth $W$ induced by the market.
As a quick warm-up, when we say $W_1$ is \textit{preferred} to $W_2$, we write $\mds{E}\lb u(W_1) \rb \geq \mds{E}\lb u(W_2) \rb$.

\begin{defn}[Risk Aversion \cite{Back2017asset}]
An investor is said to be \textbf{risk averse} if $u(\mu_W) \ge \mds{E}[u(W)]$ for every random variable $W$ with mean $\mu_W$.
Equivalently, $u(\mu_W) \ge \mds{E}[u(\mu_W+Z)]$ for every zero-mean random variable $Z$.
\end{defn}

Since risk aversion is characterized by Jensen's inequality, it naturally corresponds to the
\textit{concavity}\footnote{Intuitively, utility exhibits diminishing marginal utility of wealth: as wealth grows, each additional unit of payoff contributes less utility. Hence a fair gamble can lower expected utility, because the utility loss in bad states outweighs the utility gain in good states.}
of the utility function.
To quantify the degree of concavity of $u$, we introduce the coefficients of risk aversion, which play a role analogous to curvature.

\begin{defn}[Coefficients of Risk Aversion \cite{Back2017asset}]\label{def:CRRA}
Given a wealth level $w$ and a utility function $w \mapsto u(w)$, the \textbf{coefficient of absolute risk aversion} is defined as
\begin{equation} \label{eq:abs_risk_aver}
\alpha_u(w) := -\frac{u''(w)}{u'(w)},
\end{equation}
while the \textbf{coefficient of relative risk aversion} is defined as
\begin{equation} \label{eq:rel_risk_aver}
\rho_u(w) := w\alpha_u(w) = -\frac{w u''(w)}{u'(w)}.
\end{equation}
\end{defn}
Suppose that an investor chooses a portfolio vector
\(\mb b=(b_1,\ldots,b_m)\) from the simplex \(\Delta^{m-1}\), and let
\(\mb X=(X_1,\ldots,X_m)\) denote the nonnegative random payoff vector of the \(m\) assets.
The resulting wealth is
\begin{equation}
    W:=\braket{\mb b,\mb X}.
\end{equation}
Following the finite-state portfolio-choice setting in the financial literature \cite{Back2017asset}, we assume that \(\mb X\) takes values in a finite set \(\mcal X\), with support
\(\supp(p_{\mb X})=\{\mb x_1,\ldots,\mb x_k\}\subseteq\mcal X\).
Throughout the paper, we restrict attention to portfolios satisfying
\(\braket{\mb b,\mb x}>0\) for all \(\mb x\in\supp(p_{\mb X})\).

\begin{defn}[Payoff Matrix \cite{Back2017asset}] \label{def:payoff_matrix}
Let \(\mcal{X}\) be the state space of the market payoff vector, and suppose that $\supp(p_{\mb{X}})=\{\mb{x}_1,\dots,\mb{x}_k\}\subseteq \mcal{X}$.
The \textbf{payoff matrix} is defined by
\begin{equation}
    M_X :=
    \begin{pmatrix}
        \mb{x}_1^{\mathsf T}\\
        \vdots\\
        \mb{x}_k^{\mathsf T}
    \end{pmatrix}
    \in \mathbb{R}_+^{k\times m}.
\end{equation}
Thus, for any portfolio $\mb{b}\in\Delta^{m-1}$, the wealth vector realized in state $\mb{x}_j$ is
$(M_X\mb{b})_j := \braket{\mb{b}, \mb{x}_j}$ for each $j\in[k]$.
\end{defn}

To model a specific investor, we assume that the utility function $u$ has the same relative risk-aversion coefficient $\rho_u$ at each wealth level $w$.
In fact, for this standard modeling \cite{Mossin1968MultiPeriod,Back2017asset,Campbell2017Financial}, the investor's utility function $u(w)$ admits a closed form as follows.

\begin{prop}[CRRA Utility {\cite{Mossin1968MultiPeriod}}]
The constant relative risk aversion (CRRA) utility function is given by any affine transform of
\begin{subnumcases}
{u(w) = }
\frac{1}{1-\rho_u} w^{1-\rho_u},& $\rho_u \in\,(0, 1) \cup (1, \infty);$\\
\log w,& $\rho_u = 1$.
\end{subnumcases}
\end{prop}

With this modeling of portfolio selection, the investor's optimization problem becomes
\begin{equation} \label{eq:max_CRRA}
    \max_{\mb{b}\in \mathbb{R}^m} \mds{E}_p[u(\braket{\mb{b}, \mb{X}})]\quad \text{ subject to } \sum_{i=1}^m b_i = 1 \text{ and } b_i \geq 0,
\end{equation}
for every asset $i \in [1:m]$.

Next, to build a bridge connecting the portfolio selection problem to information projection problem, we introduce the concept of \textit{certainty-equivalent growth rate}, which should be regarded as an extension of the utility function in gambling \cite[Eq.~(13)]{Bleuler2020bet}. 

\begin{defn}[Certainty-Equivalent Growth Rate] \label{def:CE_growth}
For a portfolio vector $\mb{b}$ and a relative risk aversion coefficient $\rho_u \in (0,\infty)$, the \textbf{Certainty-Equivalent (CE) Growth Rate} is defined as:
\begin{subnumcases}
{G_{\rho_u}(W) := }
\frac{1}{1-\rho_u} \log \mds{E}_p \lb \braket{\mb{b}, \mb{X}}^{1-\rho_u} \rb, & $\rho_u \neq 1;$\\
\mds{E}_p\lb\log W\rb, & $\rho_u = 1;$
\end{subnumcases}
\end{defn}
We set $G_1(W):=\mds{E}_p[\log W]$ by continuity.
\begin{remark}
\label{rmk:utility_growth_equiv}
For CRRA utility, \(G_{\rho_u}(W)\) is the logarithm of the certainty-equivalent wealth:
\begin{equation}
w_{\mrm{CE}} := u_{\rho_u}^{-1}\!\left(\mds{E}_p[u(W)]\right) =
\exp\!\left(G_{\rho_u}(W)\right).
\end{equation}
Equivalently, for \(\rho_u\neq1\),
\begin{equation}
\mds{E}_p[W^{1-\rho_u}]=\exp((1-\rho_u)G_{\rho_u}(W)).    
\end{equation}
Notably, the preference order for the investor is preserved.
Therefore, maximizing expected CRRA utility is equivalent to maximizing \(G_{\rho_u}(W)\) for every \(\rho_u>0\), where the logarithmic case ($\rho_u = 1$) is obtained by continuity.
\end{remark}

\begin{defn}[Symmetric Covering Set]
\label{def:X_cube}
Let \(\mcal{X}\subseteq\mds{R}_+^m\) be a finite payoff reference set.
A finite set \(\mcal{X}^{\square}\supseteq\mcal{X}\) is called a \textbf{symmetric covering set} if there exists \(\gamma\in\mds{R}\) such that
\begin{equation}
\sum_{\mb{x}\in\mcal{X}^{\square}}\mb{x}
=
\gamma\mds{1}.
\label{eq:symmetric_covering_sum}
\end{equation}
Equivalently, when the vectors in \(\mcal X^\square\) are summed coordinate-wise, each asset coordinate has the same total.
\end{defn}

One concrete construction is obtained by reflection.
Let \(M_X\in\mds{R}_+^{k\times m}\) be the payoff matrix in Definition~\ref{def:payoff_matrix}, and define
\(a_{\max}:=\max_{j\in[k],\,i\in[m]}(M_X)_{ji}\).
Then \[\mcal{X}^{\square} := \mcal{X}\cup\{a_{\max}\mds{1}-\mb{x}:\mb{x}\in\mcal{X}\}\] is a symmetric covering set. Indeed, the reflection
\(\mb{x}\mapsto a_{\max}\mds{1}-\mb{x}\) is a bijection on \(\mcal{X}^{\square}\), so we have
\[\sum_{\mb{x}\in\mcal{X}^{\square}}\mb{x} = \frac{a_{\max}|\mcal{X}^{\square}|}{2}\mds{1}.\]

To compare the market law with the portfolio-induced wealth law, we regard both objects as measures on the common finite reference domain \(\mcal X^\square\).
This domain is a modeling choice and need not coincide with the realized support \(\supp(p_{\mb X})\).
Throughout this section, \(p_{\mb X}\) is extended to \(\mcal X^\square\) by setting \(p(\mb x)=0\) outside \(\supp(p_{\mb X})\).
We now define the risk-tilted market law and the portfolio-induced wealth law on this common domain.

\begin{defn}[Tilted Measure \cite{Verdu2015alpha}] \label{def:tilted_p}
For the probability law \(p\) of a random vector \(\mb{X}\), we define its \textbf{tilted measure} of order \(\beta\) on the finite set $\mcal{A}$ by
\begin{equation}
    \tilde{p}_{\beta}(\mb{x}) := \frac{p(\mb{x})^{\beta}}{Z_p},
\end{equation}
where \(Z_p := \sum_{\mb{x} \in \mcal{A}}p(\mb{x})^{\beta}\).
\end{defn}

Throughout this work, we mainly focus on the case $\beta = \frac{1}{\rho_u}$, and denote the \textit{risk-tilted measure} as $\tilde{p} := \tilde{p}_{\frac{1}{\rho_u}}$ on a finite set \(\mcal{A} = \mcal X^\square\).

\begin{defn}[Portfolio-Induced Measure] \label{def:q_b}
Given a portfolio \(\mb b\in\Delta^{m-1}\), define the induced \textbf{wealth measure} on \(\mcal X^\square\) by
\(q_{\mb b}(\mb x):=\braket{\mb b,\mb x}\).
The corresponding normalized probability measure is
\begin{equation}
    \bar q_{\mb b}(\mb x):=\frac{\braket{\mb b,\mb x}}{Z_q},
\end{equation}
where the (wealth) partition function is \(Z_q:=\sum_{\mb{x}\in\mcal X^\square}\braket{\mb b,\mb x}\).
\end{defn}

\begin{exam} \label{ex:linear_bin}
Let \(\mcal X=\{0,c\}^m\), where \(c>0\).
Then \(\mcal X\) is already a symmetric covering set, so we may take \(\mcal X^\square=\mcal X\).
The partition function $Z_q$ can be calculated as
\begin{align*}
&Z_q = \sum_{\mb{x}\in\mathcal X^\square} \braket{\mb{b},\mb{x}} = c \cdot 2^{m-1}.
\end{align*}
\end{exam}

The full calculation is given in Appendix~\ref{app:b_indep_pf}.
Example~\ref{ex:linear_bin} illustrates why symmetrization is useful:
when the reference set is symmetric, the normalizing constant $Z_q$ does not depend on the portfolio \(\mb b\).
Lemma~\ref{lemma:b_indep} extends this observation to general symmetric covering sets.

\begin{lemma}[Portfolio-Independence of the Partition Function]
\label{lemma:b_indep}
If \(\mcal{X}^{\square}\) is a symmetric covering set, then the partition function
\begin{equation*}
Z_q = \sum_{ \mb{x} \in \mcal{X}^{\square}} \langle\mb{b},\mb{x}\rangle
\end{equation*}
is independent of the portfolio choice \(\mb{b}\in\Delta^{m-1}\).
\end{lemma}
\begin{proof}
Since \(\mcal{X}^{\square}\) is a symmetric covering set, there exists \(\gamma\in\mds{R}\) such that
\begin{equation}
\sum_{\mb{x}\in\mcal{X}^{\square}}\mb{x} = \gamma\mds{1}.
\end{equation}
Therefore, using linearity of the inner product,
\begin{align}
Z_q = \sum_{\mb{x}\in\mcal{X}^{\square}}
\langle\mb{b},\mb{x}\rangle = \left\langle \mb{b}, \sum_{\mb{x}\in\mcal{X}^{\square}}\mb{x} \right\rangle =
\langle\mb{b},\gamma\mds{1}\rangle =
\gamma,
\end{align}
where the last equality follows from the simplex constraint
\(\langle\mb{b},\mds{1}\rangle=1\).
Thus \(Z_q\) depends only on the symmetric covering set and is independent of \(\mb{b}\).
\end{proof}

\begin{defn}[Rényi Divergence and Rényi Entropy \cite{Renyi1961Entropy}] \label{def:infor_measure}
Let $p$ and $q$ be probability distributions on a finite alphabet $\mcal{A}$, and let $\alpha \in (0,1)\cup(1,\infty)$.
The \textbf{Rényi divergence} of order $\alpha$ from $p$ to $q$ is defined as 
\begin{equation}
    \rD{\alpha}{p}{q} := \frac{1}{\alpha-1} \log \sum_{x \in \mcal{A}} p(x)^\alpha q(x)^{1-\alpha}. 
\end{equation}
The \textbf{Rényi entropy} of order $\alpha$ is defined as 
\begin{equation}
    \rET{\alpha}{p} := \frac{1}{1-\alpha}
\log \sum_{x \in \mcal{A}} p(x)^\alpha.
\end{equation}
By continuity at $\alpha=1$, these reduce to the Kullback--Leibler divergence and the Shannon entropy, respectively.
\end{defn}

\section{Main Results}
With the single-period CRRA portfolio selection problem formulated, now we characterize the CE growth rate (Definition~\ref{def:CE_growth}) in terms of information measures (Definition~\ref{def:infor_measure}).
\begin{theo}[Characterization of CE Growth Rate] \label{thm:G_rho}
For any coefficient of relative risk aversion $\rho_u \in (0, \infty)$, the CE growth rate can be expressed as 
\begin{align}
G_{\rho_u}(W) = - \rD{\rho_u}{\tilde{p}}{\bar{q}_\mb{b}} - \rET{\rho_u}{\tilde{p}} + \log Z_q.
\end{align}
\end{theo}

\begin{remark}[Interpretation of the first term $\rD{\rho_u}{\tilde{p}}{\bar{q}_{\mathbf{b}}}$ in Theorem~\ref{thm:G_rho}]
The Rényi divergence term measures the mismatch between the investor's risk-tilted belief $\tilde{p}$ and the portfolio-induced wealth distribution $\bar{q}_{\mathbf{b}}$.
Notably, its order $\rho_u$ is exactly the investor's coefficient of relative risk aversion.
\end{remark}

\begin{remark}[Interpretation of the second term $\rET{\rho_u}{\tilde{p}}$ in Theorem~\ref{thm:G_rho}]
The Rényi entropy term quantifies the uncertainty of the risk-tilted law $\tilde{p}$.
It is independent of the portfolio $\mathbf{b}$ and depends only on the market law through the investor's tilted belief.
This implies that the investor's higher uncertainty about the market will lower $G_{\rho_u}(W)$ and hence the expected CRRA utility.
\end{remark}

\begin{proof}[Proof of Theorem~\ref{thm:G_rho}]
For $\rho_u \neq 1$, we rewrite CE growth rate as
\begin{align}
    G_{\rho_u}(W)
    =& \frac{1}{1-\rho_u} \log \sum_{\mb{x} \in {\mcal{X}}^{\square}} \lp p(\mb{x})^{\frac{1}{\rho_u}}\rp^{\rho_u} \braket{\mb{b}, \mb{x}}^{1-\rho_u}\notag\\ &- \frac{1}{1-\rho_u} \log \sum_{\mb{x} \in {\mcal{X}}^{\square}} \lp p(\mb{x})^{\frac{1}{\rho_u}}\rp^{\rho_u}\\
    =& \frac{1}{1-\rho_u} \log \frac{\sum_{\mb{x} \in {\mcal{X}}^{\square}} \lp \frac{p(\mb{x})^{\frac{1}{\rho_u}}}{Z_p}\rp^{\rho_u} \braket{\mb{b}, \mb{x}}^{1-\rho_u}}{\sum_{\mb{x} \in {\mcal{X}}^{\square}} \lp \frac{p(\mb{x})^{\frac{1}{\rho_u}}}{Z_p}\rp^{\rho_u}}\\
    =& \frac{1}{1-\rho_u} \log \frac{\sum_{\mb{x} \in {\mcal{X}}^{\square}} \tilde{p}(\mb{x})^{\rho_u} \bar{q}_{\mb{b}}(\mb{x})^{1-\rho_u}}{\sum_{\mb{x} \in {\mcal{X}}^{\square}} \tilde{p}(\mb{x})^{\rho_u}}  + \log Z_q\\
    =& -\rD{\rho_u}{\tilde{p}}{\bar{q}_\mb{b}} - \rET{\rho_u}{\tilde{p}} + \log Z_q,
\end{align}
and the $\rho_u = 1$ case holds by the continuity of Rényi divergence \cite[Thm.\,7]{Van2014Renyi} and Rényi entropy \cite{Renyi1961Entropy}. 
\end{proof}

By the non-negativity of Rényi divergence \cite[Thm.\,8]{Van2014Renyi}, Corollary~\ref{cor:EU_upper_bnd} provides an upper bound for the expected CRRA utility.
Its tightness is governed by the projection error: the optimization projects the risk-tilted law $\tilde p$ onto the portfolio-induced linear family $\{\bar q_{\mb b}:\mb b\in\Delta^{m-1}\}$, and equality holds precisely when $\tilde p$ already lies in that family.
\begin{coro} \label{cor:EU_upper_bnd} 
For $\rho_u \in (0,1) \cup (1,\infty)$,
\begin{equation}
    \mds{E}_{p}\lb u(W)\rb \leq \frac{1}{1-\rho_u} \e^{(1-\rho_u) \lp - \rET{\rho_u}{\tilde{p}} + \log Z_q \rp};
\end{equation}
for $\rho_u = 1$,
\begin{equation}
    \mds{E}_{p}\lb u(W)\rb = \mds{E}_{p}\lb \log W\rb \leq - \ET{p} + \log Z_q.
\end{equation}
\end{coro}

Before deriving an algorithm, we first convert the CE-growth decomposition established in Theorem~\ref{thm:G_rho} into an equivalent projection problem, where we project the risk-tilted law \(\tilde p\) onto the portfolio-induced family \(\{\bar q_{\mb b}:\mb b\in\Delta^{m-1}\}\) under Rényi divergence.

\begin{theo} \label{thm:renyi_projection}
If the random payoff vector $\mb{X}$ has finite support $\mcal{X}$, then for all $\rho_u \in (0, \infty)$, 
\begin{align}
\argmax_{\mb{b}\in\Delta^{m-1}} \mds{E}_{p}[u(W)] = \argmin_{\mb{b}\in\Delta^{m-1}}
\rD{\rho_u}{\tilde p}{\bar q_{\mb b}}.
\end{align}
\end{theo}

\begin{proof}
From Remark~\ref{rmk:utility_growth_equiv}, we observe that maximizing expected CRRA utility is equivalent to maximizing the CE growth rate $G_{\rho_u}$.
Also, by Lemma~\ref{lemma:b_indep} we know that \(Z_q\) is independent of the portfolio choice $\mb b$ after identifying a symmetric covering set.
Therefore, maximizing $G_{\rho_u}$ is equivalent to minimizing \(\rD{\rho_u}{\tilde p}{\bar q_{\mb b}}\). 
Finally, this reduces  to the KL projection \(\KLD{p}{\bar q_{\mb b}}\) at \(\rho_u=1\) by continuity.
\end{proof}

\begin{prop}[A variational formula for R\'enyi divergence {\cite[Thm.\,1]{Shayevitz2010Note} \cite[Thm.\,30]{Van2014Renyi} \cite[Lem.\,2.1]{Esposito2025Sibson}}] 
\label{prop:var_renyi}
Let \(p\) and \(q\) be probability distributions on a finite alphabet \(\mcal A\), and let \(\alpha\in(0,1)\cup(1,\infty)\).
Then we have
\begin{equation*}
(1-\alpha)\rD{\alpha}{p}{q} = \min_{r\in\Delta(\mcal A)} \left\{ \alpha\KLD{r}{p}+(1-\alpha)\KLD{r}{q} \right\}.
\end{equation*}
 The unique minimizer is
\[r^\star(x) = \frac{p(x)^\alpha q(x)^{1-\alpha}}{Z_\alpha(p,q)}\]
with the normalizer \[ Z_\alpha(p,q):=\sum_{x\in\mcal A}p(x)^\alpha q(x)^{1-\alpha} \in (0,\infty).\]
\end{prop}

The information projection view becomes more algorithmic once the Rényi divergence is replaced by its variational representation.
Applying Proposition~\ref{prop:var_renyi} to Theorem~\ref{thm:renyi_projection} gives the following Blahut--Arimoto-style alternating optimization.
\begin{theo}[Alternating Optimization for Expected CRRA Utility] \label{thm:renyi_variational}
\begin{align*}
&\quad\argmax_{\mb{b} \in \Delta^{m-1}} \mds{E}_{p}\lb u(W)\rb
= \argmax_{\mb{b} \in \Delta^{m-1}} G_{\rho_u}(W) = \\
&\left\{
\begin{aligned}
&\argmin_{\mb{b} \in \Delta^{m-1}} \max_{r \in \Delta(\mathcal X^\square)}
   \lbp
   - \frac{\rho_u}{\rho_u - 1} \KLD{r}{\tilde{p}}
   + \KLD{r}{\bar{q}_{\mb{b}}}
   \rbp,
&& \rho_u > 1;\\
&\argmin_{\mb{b} \in \Delta^{m-1}} \KLD{p}{\bar{q}_{\mb{b}}},
&& \rho_u = 1;\\
&\argmin_{\mb{b} \in \Delta^{m-1}} \min_{r \in \Delta(\mathcal X^\square)}
   \lbp
   \frac{\rho_u}{1-\rho_u} \KLD{r}{\tilde{p}}
   + \KLD{r}{\bar{q}_{\mb{b}}}
   \rbp,
&& \rho_u < 1.
\end{aligned}
\right.
\end{align*}
For each fixed portfolio $\mb{b}$, the optimizer $r^\star$ has the closed-form expression
$r^\star(\mb{x}) \propto \tilde{p}(\mb{x})^{\rho_u}\bar{q}_\mb{b}(\mb{x})^{1-\rho_u}$.
\end{theo}
Theorem~\ref{thm:renyi_variational} converts the Rényi information projection problem (Theorem~\ref{thm:renyi_projection}) into a Blahut--Arimoto-style alternating optimization through KL-divergence information projection.
The auxiliary state law $r$ has a closed-form update, and the portfolio block for $\mb b$ is a convex subproblem.

\section{Numerical Experiments}
\label{sec:num}

\begin{figure*}[t]
\centering
\includegraphics[width=0.95\textwidth]{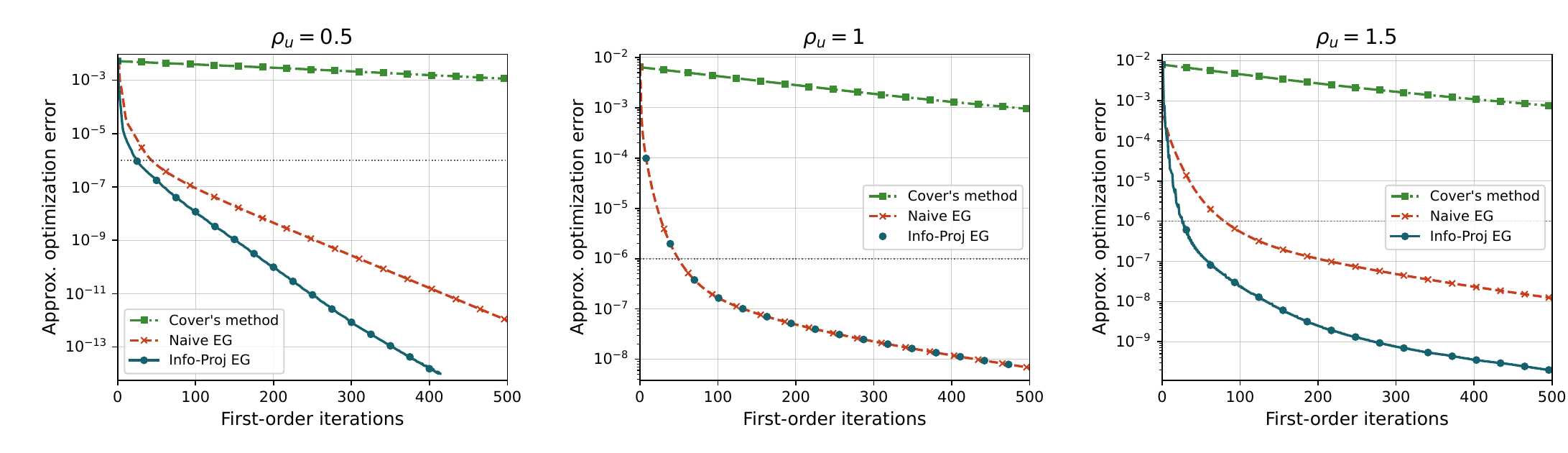}
\caption{Numerical comparison for CRRA portfolio selection.}
\label{fig:dirichlet_portfolio}
\end{figure*}

We implement the alternating optimization suggested by Theorem~\ref{thm:renyi_variational}.
The variational characterization is valid for all \(\rho_u>0\), but our numerical study focuses on \(\rho_u\in(0,2)\), where the implementation is supported by analysis: alternating-minimization descent for \(\rho_u\in(0,1)\), reduction to the expected-log problem at \(\rho_u=1\), and local contractivity of the alternating map for \(\rho_u\in(1,2)\).
We defer the precise update rules and supporting analyses to Appendix~\ref{app:num-analysis}.
For the \(\rho_u \ge 2\) regime, developing efficient implementations of the Rényi information projection (Theorem~\ref{thm:renyi_projection}) or its variational formulation (Theorem~\ref{thm:renyi_variational}) is left for future work.

The resulting method, called \emph{Info-Proj EG}, uses the closed-form auxiliary-law update and then performs exponentiated-gradient (EG) steps with Armijo backtracking on the fixed-\(r\) logarithmic portfolio objective~\cite{Armijo1966Lipschitz,LiCevher2019EGArmijo}.
We compare it with two baselines.
The first, called \emph{Naive EG}, directly minimizes the CRRA-equivalent loss.
The second, labeled \emph{Cover's method} in Fig.~\ref{fig:dirichlet_portfolio}, is Cover's multiplicative update at \(\rho_u=1\) and the corresponding FOC-induced fixed-point analogue for \(\rho_u\neq1\).

For \(j\in[k]\), let \(q_j(\mb b):=(M_X\mb b)_j\).
Info-Proj EG calculates
\begin{equation}
r_j(\mb{b}) \propto p_jq_j(\mb{b})^{1-\rho_u},
\end{equation}
and then updates the portfolio by minimizing
\begin{equation} \label{eq:infor_proj_obj}
f_r(\mb{b})=-\sum_{j=1}^k r_j\log q_j(\mb{b}).
\end{equation}
Naive EG minimizes
\begin{equation} \label{eq:naive_obj}
F_{\rho_u}^{\rm naive}(\mb{b})=
\begin{cases}
-\sum_{j=1}^k p_j q_j(\mb{b})^{1-\rho_u}, & 0<\rho_u<1,\\
-\sum_{j=1}^k p_j\log q_j(\mb{b}), & \rho_u=1,\\
\sum_{j=1}^k p_j q_j(\mb{b})^{1-\rho_u}, & \rho_u>1.
\end{cases}
\end{equation}
For each asset \(i\in[m]\), Cover's FOC-induced fixed-point update is
\begin{equation}\label{eq:cover_update}
b_i^{(t+1)} \leftarrow b_i^{(t)}\frac{\left(\nabla_{\mb b}\mathbb{E}_{p}\!\left[u\!\lp \langle \mb b,\mb X\rangle \rp\right](\mb b^{(t)})\right)_i}{\left\langle \mb b^{(t)},\nabla_{\mb b}\mathbb{E}_{p}\!\left[u\!\lp \langle \mb b,\mb X\rangle \rp\right](\mb b^{(t)})\right\rangle}.
\end{equation}

For \(\rho_u \in (0,1)\), Theorem~\ref{thm:renyi_variational} gives an alternating-minimization formulation, with a closed-form \(r\)-update and an EG/Armijo portfolio step.
For \(\rho_u>1\), the variational form becomes min--max, so the corresponding comparison is empirical and does not imply a uniform advantage over Naive EG.
Additional numerical detail is given in Appendix~\ref{app:num-analysis}.

Fig.~\ref{fig:dirichlet_portfolio} is computed on a \((k,m)=(100,50)\) payoff matrix, with state probabilities \(p\) drawn from \(\operatorname{Dir}(10\cdot\mds{1}_k)\).
In this instance, Info-Proj EG reaches smaller optimization error in fewer iterations than the two baselines for \(\rho_u\in\{0.5,1.5\}\), while matching Naive EG at \(\rho_u=1\), where both EG methods optimize the same expected log-wealth objective.

\section{Discussion}

A natural modeling question is how many assumptions one should impose on the market model for tractability.
Since Samuelson~\cite{Samuelson1969lifetime} and Merton~\cite{Merton1969lifetime}, portfolio selection has often been studied through dynamic formulations with intertemporal market dynamics or structural assumptions.
Later work further incorporates predictive state variables, parametric return dynamics, and trading frictions~\cite{AitSahaliaBrandt2001Variable,CampbellEtAl2004VAR,MaSiuZhu2019ReturnPredictabilityTC}.
For a broad modern account of the computational side of portfolio optimization, see~\cite{Palomar2025portfolio}.
By contrast, this work assumes only finite support and asks what information-theoretic structure is already present in the single-period CRRA portfolio selection problem.

Following Markowitz's classical two-step viewpoint in portfolio theory~\cite{Markowitz1952Portfolio}, one first estimates a probability distribution \(p\) for the market payoff vector \(\mb X\), and then chooses a portfolio \(\mb b\) accordingly.
A key limitation of this work is that we focus only on the second step: selecting \(\mb b\) given a payoff law, while leaving the estimation of \(p\) outside the scope of the paper.

Our formulation should be distinguished from Online Portfolio Selection (OPS)~\cite[Ch.\,13]{Orabona2019ModernOnline}, where no market law is assumed and performance is measured relative to competing experts.
Classical milestones of OPS include Cover's universal portfolio~\cite{Cover1991universal}, \cite[Ch.\,16]{Cover2006}, the sequential-learning viewpoint~\cite{Merhav1992sequential}, and multiplicative-update methods~\cite{Helmbold1998OPSMultiplicative}.
Recent advances continue to sharpen this line through data-dependent regret guarantees~\cite{Tsai2023data} and efficient near-optimal algorithms~\cite{Jezequel2025Efficient}.

The objective in this work is the expected CRRA utility, in contrast to the mean--variance paradigm of Markowitz~\cite{Markowitz1952Portfolio} and Sharpe~\cite{Sharpe1964Risk}, which mainly considers the first and second moments.
This mean--variance approach is useful in data-limited regimes, but higher-order moments matter when returns exhibit asymmetry or heavy tails~\cite{Harvey2000ConditionalSkewness,Jondeau2006HigherMoments,Cvitanic2008HigherMoments}.

One possible future direction is \textit{estimation-aware information projection}.
In practice, the market law \(p\) is usually just an estimate.
It would be useful to understand how sensitive the optimal portfolio \(\mb b^\star(p)\) is to perturbations of \(p\), and to develop portfolio-selection policies that project not only a single empirical law, but a neighborhood accounting for estimation error around it.


\clearpage \balance
\bibliographystyle{IEEEtran}
\bibliography{ref}


\newpage
\nobalance
\appendices
\onecolumn
\section{Calculation in Example~\ref{ex:linear_bin}}
\label{app:b_indep_pf}


Specifically, the partition function $Z_q$ in Example~\ref{ex:linear_bin} can be calculated as
\begin{align*}
&Z_q = \sum_{\mb{x}} \braket{\mb{b},\mb{x}} \overset{(a)}{=} c\sum_{\mb{\bar{x}}} \braket{\mb{b},\mb{\bar{x}}} \overset{(b)}{=} c \cdot 2^m \mds{E}_{\mb{\bar{X}} \overset{\mrm{i.i.d.}}{\sim}\mrm{Ber}(\frac{1}{2})} \lb \braket{\mb{b},\mb{\bar{X}}}\rb\\
&= c \cdot 2^m \braket{\mb{b},\mds{E}_{\mb{\bar{X}} \overset{\mrm{i.i.d.}}{\sim}\mrm{Ber}(\frac{1}{2})} \lb \mb{\bar{X}}\rb} = c \cdot 2^m \sum_{i=1}^m b_i \cdot \frac{1}{2} = c \cdot 2^{m-1},
\end{align*}
where $(a)$ follows from normalizing $\mb{x}$ to $\mb{\bar{x}}$ such that $\bar{x}_i \in \{0, 1\}$;
$(b)$ follows from the symmetry of binary support for each $\mb{\bar{x}}_i$, namely, the number of appearances of \(1\) equals that of \(0\) across all binary vectors \(\bar x\).

\section{Numerical Analysis}
\label{app:num-analysis}

This appendix records the numerical update rules used in Section~\ref{sec:num} and gives two analytical interpretations of Fig.~\ref{fig:dirichlet_portfolio}.
First, in the alternating-minimization regime \(\rho_u \in (0,1)\), accepted fixed-\(r\) EG steps decrease the R{\'e}nyi projection objective.
Second, the alternating update is locally contractive near a regular interior fixed point for \(\rho_u \in (0,2)\).
For \(j\in[k]\), let \(q_j(\mb b):=(M_X\mb b)_j\).

\subsection{Implemented Update Rules}

For a differentiable objective \(F\) on the simplex, exponentiated gradient (EG) is the KL-proximal mirror-descent step
\begin{equation}
\mb b^+=\mcal{E}_{\eta,F}(\mb b):=\argmin_{\mb u\in\Delta^{m-1}}\left\{\eta\langle\nabla F(\mb b),\mb u-\mb b\rangle+\KLD{\mb u}{\mb b}\right\}.
\end{equation}
Equivalently,
\begin{equation}
\left(\mcal{E}_{\eta,F}(\mb b)\right)_i=\frac{b_i\exp\!\left(-\eta(\nabla F(\mb b))_i\right)}{\sum_{\ell=1}^m b_\ell\exp\!\left(-\eta(\nabla F(\mb b))_\ell\right)}, \quad i\in[m].
\end{equation}
Armijo backtracking selects the stepsize along this EG path.
Given \(c\in(0,1)\), the candidate is accepted once
\begin{equation}
F(\mcal{E}_{\eta,F}(\mb b))\le F(\mb b)+c\langle\nabla F(\mb b),\mcal{E}_{\eta,F}(\mb b)-\mb b\rangle.
\end{equation}
The line-search rule follows Armijo~\cite{Armijo1966Lipschitz}; convergence of EG/Armijo on the simplex for a fixed convex differentiable objective is studied in~\cite{LiCevher2019EGArmijo}.
In Fig.~\ref{fig:dirichlet_portfolio}, one first-order iteration means one accepted EG/Armijo step for Naive EG and Info-Proj EG, and one multiplicative fixed-point update for the curve labeled \emph{Cover's method}.

The method suggested by Theorem~\ref{thm:renyi_variational}, called Info-Proj EG, sets the auxiliary law by
\begin{equation}
r_j(\mb b)=\frac{p_jq_j(\mb b)^{1-\rho_u}}{\sum_{\ell=1}^k p_\ell q_\ell(\mb b)^{1-\rho_u}}, \quad j\in[k].
\end{equation}
It then applies EG/Armijo to the fixed-\(r\) logarithmic portfolio objective
\begin{equation}
f_r(\mb b):=-\sum_{j=1}^k r_j\log q_j(\mb b).
\end{equation}
At \(\rho_u=1\), the auxiliary law reduces to \(r(\mb b)=p\).

The direct baseline, called Naive EG, applies EG/Armijo to the CRRA-equivalent loss
\begin{equation}
F_{\rho_u}^{\rm naive}(\mb b)=\begin{cases} -\sum_{j=1}^k p_jq_j(\mb b)^{1-\rho_u}, & 0<\rho_u<1,\\ -\sum_{j=1}^k p_j\log q_j(\mb b), & \rho_u=1,\\ \sum_{j=1}^k p_jq_j(\mb b)^{1-\rho_u}, & \rho_u>1. \end{cases}
\end{equation}

The curve labeled \emph{Cover's method} in Fig.~\ref{fig:dirichlet_portfolio} is generated as follows.
For \(\rho_u=1\), it is Cover's multiplicative update in~\eqref{eq:cover_update}.
For \(\rho_u\neq1\), it is the fixed-point update suggested by the first-order condition~\eqref{eq:crra_foc}:
\begin{equation}
b_i^{(t+1)} \leftarrow b_i^{(t)}\frac{\sum_{j=1}^k p_j(M_X)_{ji}q_j(\mb b^{(t)})^{-\rho_u}}{\sum_{j=1}^k p_jq_j(\mb b^{(t)})^{1-\rho_u}}, \quad i\in[m].
\end{equation}
This update preserves the simplex by normalization.
At \(\rho_u=1\), the denominator equals one, and the update reduces to Cover's logarithmic update.
For \(\rho_u\neq1\), we use it only as a fixed-point reference induced by the first-order condition.

\subsection{Descent in the Alternating-Minimization Regime for \(\rho_u \in (0,1)\)}

For \(\rho_u \in (0,1)\), Theorem~\ref{thm:renyi_variational} gives the alternating-minimization formulation
\begin{equation}
\min_{\mb b\in\Delta^{m-1}}\min_{r\in\Delta_k} J_{\rho_u}(r,\mb b),
\end{equation}
where
\begin{equation}
J_{\rho_u}(r,\mb b):=\frac{\rho_u}{1-\rho_u}\KLD{r}{\tilde p}+\KLD{r}{\bar q_{\mb b}}.
\end{equation}
For fixed \(\mb b\), the minimizer is \(r(\mb b)\).
For fixed \(r\), minimizing \(J_{\rho_u}(r,\mb b)\) over \(\mb b\) is equivalent, up to constants independent of \(\mb b\), to minimizing \(f_r(\mb b)\).
Hence any accepted EG step that decreases \(f_{r^{(t)}}\) also decreases the R{\'e}nyi projection objective.
Indeed, if \(r^{(t)}=r(\mb b^{(t)})\) and \(f_{r^{(t)}}(\mb b^{(t+1)})\le f_{r^{(t)}}(\mb b^{(t)})\), then
\begin{align}
\rD{\rho_u}{\tilde p}{\bar q_{\mb b^{(t+1)}}}
&=\min_{r\in\Delta_k} J_{\rho_u}(r,\mb b^{(t+1)}) \le J_{\rho_u}(r^{(t)},\mb b^{(t+1)}) \notag\\
&\le J_{\rho_u}(r^{(t)},\mb b^{(t)}) =\rD{\rho_u}{\tilde p}{\bar q_{\mb b^{(t)}}}.
\end{align}
This is the objective-level descent property available in the low-risk-aversion \((\rho_u < 1)\) regime.

At \(\rho_u=1\), the auxiliary law reduces to \(r=p\), and the fixed-\(r\) logarithmic objective coincides with the Naive EG objective.
Thus Info-Proj EG and Naive EG solve the same fixed-law log-utility allocation problem.
With identical EG/Armijo scheduling, their iterates coincide, as reflected in Fig.~\ref{fig:dirichlet_portfolio}.

For \(\rho_u>1\), the same variational representation becomes
\begin{equation}
\min_{\mb b\in\Delta^{m-1}}\max_{r\in\Delta_k}\left\{-\frac{\rho_u}{\rho_u-1}\KLD{r}{\tilde p}+\KLD{r}{\bar q_{\mb b}}\right\}.
\end{equation}
A decrease of the fixed-\(r\) logarithmic surrogate therefore does not certify descent of the R{\'e}nyi projection objective after the maximizing auxiliary law is updated.
The \(\rho_u>1\) curves are empirical and do not constitute a uniform dominance claim over Naive EG.
The next subsection explains why the alternating formulation can nevertheless remain locally contractive in the range \(\rho_u \in (1,2)\).

\subsection{Local Contractivity of the Alternating Map for \(\rho_u\in(0,2)\)}

The descent argument above applies directly to the alternating-minimization regime \(\rho_u\in(0,1)\).
For \(\rho_u>1\), the variational form becomes min--max, so decreasing a fixed-\(r\) surrogate no longer certifies descent of the full Rényi projection objective.
Nevertheless, the exact alternating map has a simple local fixed-point structure.

Define
\begin{equation}
T(\mb b):=\argmin_{\mb u\in\Delta^{m-1}} f_{r(\mb b)}(\mb u),
\end{equation}
where \(r(\mb b)\) is the closed-form auxiliary law and \(f_r\) is the fixed-\(r\) logarithmic portfolio objective.
The implemented Info-Proj EG method may be viewed as an inexact realization of this map, with the fixed-\(r\) minimization carried out by EG/Armijo steps.
The result below concerns the exact map \(T\) near a regular interior fixed point.
It is therefore a local fixed-point stability statement, not a global convergence guarantee for the finite-step implementation.

\begin{prop}[Local contractivity of the alternating map]
\label{prop:local_contractivity}
Let \(\mb b^\star\in\operatorname{relint}(\Delta^{m-1})\) be a fixed point of \(T\), and set \(q^\star=M_X\mb b^\star\) and \(r^\star=r(\mb b^\star)\).
Let \(P\in\mathbb R^{m\times(m-1)}\) have columns forming an orthonormal basis of the tangent space \(\{\mb v\in\mathbb{R}^m:\mds{1}^\top\mb v=0\}\).
Assume that the fixed-\(r^\star\) logarithmic portfolio objective is strictly convex on this tangent space, namely
\begin{equation}
H:=P^\top M_X^\top \operatorname{diag}\!\left(\frac{r_j^\star}{(q_j^\star)^2}\right)M_XP\succ0.
\end{equation}
Then \(T\) is locally well defined around \(\mb b^\star\), and its tangent-space Jacobian at \(\mb b^\star\) is
\begin{equation}
J_T=(1-\rho_u)I.
\end{equation}
Consequently, for every sufficiently small tangent perturbation \(\delta\mb b\),
\begin{equation}
T(\mb b^\star+\delta\mb b)-\mb b^\star=(1-\rho_u)\delta\mb b+o(\|\delta\mb b\|).
\end{equation}
Thus the exact alternating map has local linear factor \(|1-\rho_u|\), and is locally contractive for \(\rho_u\in(0,2)\).
\end{prop}

\begin{proof}
The matrix \(H\) is the Hessian of the fixed-\(r^\star\) portfolio objective restricted to the simplex tangent space.
Thus \(H\succ0\) gives local uniqueness and differentiability of the minimizer defining \(T\), by applying the implicit-function theorem to the tangent-space first-order condition
\begin{equation}
P^\top\nabla f_{r(\mb b)}(T(\mb b))=0.
\end{equation}

Define \(D_q:=\operatorname{diag}(1/q_j^\star)\) and \(S_r:=\operatorname{diag}(r^\star)-r^\star(r^\star)^\top\).
Differentiating the auxiliary update at \(\mb b^\star\) gives
\begin{equation}
\dd r=(1-\rho_u)S_rD_qM_X\,\dd\mb b.
\end{equation}
Differentiating the tangent-space first-order condition then yields
\begin{equation}
J_T=(1-\rho_u)H^{-1}C,
\end{equation}
where
\begin{equation}
C:=P^\top M_X^\top D_qS_rD_qM_XP.
\end{equation}
It remains to show that \(C=H\).

Since \(\mb b^\star\in\operatorname{relint}(\Delta^{m-1})\) minimizes \(f_{r^\star}\) over the simplex, the KKT condition gives
\begin{equation}
M_X^\top\left(\frac{r^\star}{q^\star}\right)=\lambda\mds{1}.
\end{equation}
Multiplying by \((\mb b^\star)^\top\), we obtain
\begin{equation}
\lambda=(\mb b^\star)^\top M_X^\top\left(\frac{r^\star}{q^\star}\right)=(q^\star)^\top\left(\frac{r^\star}{q^\star}\right)=\sum_{j=1}^k r_j^\star=1.
\end{equation}
Therefore
\begin{equation}
M_X^\top\left(\frac{r^\star}{q^\star}\right)=\mds{1}.
\end{equation}
Expanding \(C\), we get
\begin{align}
C &=P^\top M_X^\top \operatorname{diag}\!\left(\frac{r_j^\star}{(q_j^\star)^2}\right)M_XP \quad -P^\top M_X^\top\left(\frac{r^\star}{q^\star}\right)\left(\frac{r^\star}{q^\star}\right)^\top M_XP.
\end{align}
The first term is \(H\), while the second term vanishes because
\begin{equation}
P^\top M_X^\top\left(\frac{r^\star}{q^\star}\right)=P^\top\mds{1}=0.
\end{equation}
Hence \(C=H\), and therefore
\begin{equation}
J_T=(1-\rho_u)I.
\end{equation}
The stated first-order expansion follows from differentiability of \(T\).
Thus the local linear factor is \(|1-\rho_u|\), which is strictly smaller than one exactly when \(\rho_u\in(0,2)\).
\end{proof}

Proposition~\ref{prop:local_contractivity} explains why the information-projection formulation can remain locally stable for \(\rho_u\in(1,2)\), even though the variational representation is min--max rather than alternating minimization.
It does not imply global convergence of the finite-step EG/Armijo implementation.

\end{document}

%% file: macro.tex
\newcommand{\lp}{\left(}
\newcommand{\rp}{\right)}
\newcommand{\lb}{\left[}
\newcommand{\rb}{\right]}
\newcommand{\lbp}{\left\{}
\newcommand{\rbp}{\right\}}

\newcommand{\mcal}{\mathcal}

\newcommand{\mb}{\mathbf}

\newcommand{\mds}{\mathds}
\newcommand{\mrm}{\mathrm}

\newcommand{\argmin}{\mathop{\mathrm{arg\,min}}}
\newcommand{\argmax}{\mathop{\mathrm{arg\,max}}}

\newcommand{\e}{\mathrm{e}}

\newcommand{\supp}{\mathsf{supp}}

\newcommand{\dd}{\mathrm{d}}

\newcommand{\ET}[1]{\mathrm{H}\mkern-1.5mu\left( #1\right)}

\newcommand{\rET}[2]{\mathrm{H}_{#1}\mkern-1.5mu\left( #2\right)}

\newcommand{\KLD}[2]{\mathrm{D}\mkern-1.5mu\left( #1  \middle\Vert #2 \right)}

\newcommand{\rD}[3]{\mathrm{D}_{#1}\mkern-1.5mu\left( #2  \middle\Vert #3 \right)}